\documentclass[11pt]{article}

\usepackage[margin=1in]{geometry}

\usepackage[T1]{fontenc}
\usepackage[utf8]{inputenc}
\usepackage{lmodern}

\usepackage{amsmath,amssymb,amsthm}

\usepackage{booktabs}
\usepackage{graphicx}
\usepackage{float}
\usepackage{needspace}

\usepackage[numbers,sort&compress]{natbib}

\usepackage[dvipsnames]{xcolor}
\usepackage[colorlinks=true,linkcolor=blue,citecolor=blue,urlcolor=blue]{hyperref}

\graphicspath{{Experiments/}}

\newcommand{\Description}[1]{}

\newif\ifeditorialnotes
\editorialnotestrue

\theoremstyle{plain}
\newtheorem{theorem}{Theorem}[section]
\newtheorem{lemma}[theorem]{Lemma}
\newtheorem{proposition}[theorem]{Proposition}
\newtheorem{corollary}[theorem]{Corollary}

\theoremstyle{definition}
\newtheorem{definition}[theorem]{Definition}

\newcommand{\Tor}{\mathcal{T}}
\newcommand{\dist}{\operatorname{dist}}

\newcommand{\Vol}{\operatorname{Vol}}
\newcommand{\gain}{\operatorname{gain}}
\newcommand{\plots}{final_allreduce_plots}

\title{SP-Allreduce: Symmetric Tree Schedules for Tori and Beyond}
\author{
Mikhail Cherniavskii,
Andrey Kupavskii,
Georgy Mironov,
Stanislav Moiseev
}
\date{}

\begin{document}

\maketitle

\begin{abstract}
Allreduce is a fundamental communication primitive in large-scale
distributed training, and torus interconnects remain widely deployed in
high-performance systems. We introduce SP-allreduce, a new allreduce algorithm for $d$-dimensional tori that exploits the symmetry of the topology and routes only along shortest paths.

We call an allreduce algorithm \emph{fully bandwidth-optimal} if it
incurs no bandwidth overhead, no multihop-induced traffic inflation, and
achieves perfect per-round load balance. SP-allreduce is based on a
family of carefully symmetrized shortest-path trees that guarantees full
bandwidth optimality in the $\alpha$--$\beta$ model. We prove that it uses the minimum possible number
of rounds among all fully bandwidth-optimal allreduce algorithms on
$d$-dimensional tori, and therefore the minimum running time in this class. In this sense, SP-allreduce refines Bucket
allreduce: it preserves Bucket's key strengths while reducing the number
of rounds by a factor of two.

We further extend our approach to broader symmetric topologies and to
communication trees of prescribed depths. Experimentally, the resulting
algorithms improve over the best compared state-of-the-art algorithms
for tori by up to $1.19\times$ on square tori and $1.65\times$ on rectangular tori
across the tested configurations. On HyperX and spine-leaf networks, we
demonstrate gains of up to $14.82\times$ and $1.72\times$, respectively,
over the best evaluated baselines.
\end{abstract}

\section{Introduction}\label{sec:intro}
The allreduce operation is a fundamental building block for both high-performance computing (HPC) and distributed deep learning. It is used to aggregate and redistribute model gradients and other global states, and is a major factor in the end-to-end runtime: prior work reports that collective operations, and allreduce in particular, may account for up to 40\% of training time in large distributed ML workloads and a substantial fraction of core-hours in production HPC systems~\citep{ben2019demystifying,chunduri2018characterization,klenk2020network,sapio2021scaling,wang2023topoopt, de2024swing}. Modern deep learning frameworks such as PyTorch expose allreduce as the main primitive for data-parallel training and heavily rely on optimized communication backends~\citep{li2020pytorch}. The importance of allreduce has motivated a long line of work on performance models, algorithmic variants, and implementations of allreduce and related collective operations.

Torus and torus-like networks are widely deployed in systems targeting both traditional HPC and machine-learning workloads. Multi-dimensional tori are used in IBM BlueGene/L and its successors~\citep{adiga2002overview}; Fujitsu's 6D Tofu-D interconnect is used in the Fugaku supercomputer,  Wisteria/BDEC-01 (Odyssey),
TOKI-SORA, several PRIMEHPC~FX1000-based systems, and Nagoya University's
Flow supercomputer~\citep{ajima2018tofu,primehpcfx1000,wisteria-bdec01,tokisora,flow}. 
 For ML workloads, vendors increasingly favor torus or torus-like fabrics: Google's TPU v5, v6 and  Ironwood systems use 2D/3D tori and optical reconfiguration to scale to thousands of accelerators~\citep{jouppi2023tpu,vahdat2023enabling, googleDocumentation, googleironwood2025}, AWS Trainium (Trn1, Trn2) uses a 2D torus for chip-to-chip connections~\citep{amazon2025trn2}, Graphcore IPU-PODs implement a dense IPU interconnect with torus structure~\citep{graphcore2025pod256}.

The allreduce collective operation combines values from all computing nodes and returns the result to every node. On general interconnects, its cost is typically modeled in the $\alpha$--$\beta$ message-cost model~\citep{rabenseifner2004optimization}. The latency term charges a fixed initiation cost $\alpha$ for each communication round; the bandwidth term is proportional to the amount of data transmitted, with $\beta$ denoting the time to transmit one bit along a link. For small messages, the data-transfer time is short, so the initiation costs of successive rounds can account for a large part of the total time. Reducing the number of rounds is therefore particularly important in this regime.

Ring allreduce~\citep{patarasuk2009bandwidth} transmits the minimum total volume, but requires $2(p-1)$ rounds for $p$ nodes. Classical recursive doubling and halving--doubling~\citep{rabenseifner2004optimization} require $O(\log p)$ communication rounds, but their communication partners may be separated by several physical links on a torus, increasing total link traffic. Different transmissions may also load the same physical links, making the distribution of traffic another factor in the bandwidth term. Section~\ref{sec:model} introduces parameters in the $\alpha$--$\beta$ model that account for communication volume, physical route lengths, and link-load imbalance.

On the one hand, the torus's limited degree and low bisection bandwidth constrain collective performance; on the other hand, its regular structure favors allreduce algorithms that are explicitly topology-aware. Indeed, there are several allreduce algorithms that were specifically designed for torus networks.
Bucket allreduce~\citep{barnett1995global,jain2010optimal,sack2015collective} is a canonical example: it preserves nearest-neighbor communication and is bandwidth-optimal on tori, but it requires a number of rounds that scales linearly with the torus diameter, which can hurt performance for small and medium message sizes.
Torus-aware algorithms such as Swing~\citep{de2024swing} and Trivance~\citep{juerss2026trivance} retain the logarithmic-round structure while reducing routing overhead through their choice of communication partners. Swing uses pairwise exchanges, whereas Trivance exchanges data with two peers at each step.

We introduce SP-allreduce, which on tori can be viewed as an improved version of Bucket. More generally, its tree-symmetrization approach applies to a broader class of topologies with sufficient symmetry.

We divide the input vector into chunks and assign them to copies of a base tree obtained using network symmetries. On equal-sided tori, translations, global reflection, and cyclic coordinate permutations form a subgroup of the topology's automorphism group that can map any directed physical link to any other. Applying these transformations to both the tree and its communication rounds distributes traffic equally across links in every round. We also adapt the construction to HyperX and spine-leaf networks.

\subsection{Our Contributions}
We introduce the SP-allreduce algorithm for torus networks and establish its bandwidth properties, analyze the number of rounds, and extend the construction:
\begin{enumerate}
\item SP-allreduce is \emph{fully bandwidth-optimal}, i.e., it incurs no bandwidth overhead, introduces no multihop-induced traffic inflation (no additional link traffic due to multi-hop routing), and achieves perfectly balanced link utilization in every round. The symmetry argument also applies to other communication trees and sufficiently symmetric topologies.
\item SP-allreduce completes in $2d\lfloor n/2\rfloor$ rounds on a $d$-dimensional $n\times\cdots\times n$ torus, approximately a factor of two fewer than Bucket~\citep{jain2010optimal}. We prove that this count is optimal among all fully bandwidth-optimal schemes for odd $n$ and exceeds the optimum by one round for even $n$. The proof derives the structure of the communication scheme from the minimum transmission volume and one-hop routing.
\item We extend the approach to communication trees of prescribed depth, introducing $k$-SP-allreduce. A tree of depth $k$ gives $k$ reduction rounds followed by $k$ dissemination rounds. Allowing a tree edge to connect nonadjacent network nodes can reduce this depth, while its physical route may increase the traffic. Experimentally, these schedules outperform the best compared torus algorithm over a wider range of message sizes.
\item Experimentally, the resulting algorithms improve over the best compared state-of-the-art algorithms for tori by up to $1.19\times$ on square tori and $1.65\times$ on rectangular tori across the tested configurations. On HyperX and spine-leaf networks, we demonstrate gains of up to $14.82\times$ and $1.72\times$, respectively, over the best evaluated baselines.
\end{enumerate}

\section{Model and Background}\label{sec:model}
\subsection{Allreduce and its representation}
Allreduce combines values from all computing nodes and returns the combined value to every node. Let $p$ be the number of nodes, each holding a vector of $N$ bits. The reduction combines corresponding vector elements using an associative and commutative operator. A common implementation consists of two stages: reduce-scatter, during which the values are collected and combined, and allgather, during which the reduced values are communicated back to every node.

We view an allreduce algorithm at two levels. At the \emph{abstract} level, an allreduce scheme specifies, for each round $s\in\{1,\ldots,S\}$, a set (or multiset) $\mathcal P_s$ of directed transmissions
\[
 u\stackrel{M}{\longrightarrow}v,
\]
meaning that node $u$ sends a message $M$ of size $|M|$ to node $v$ in round~$s$. Transmissions in a round use data available at its beginning; received values are available for subsequent rounds. The total amount of data communicated by the abstract scheme is
\[
 \Vol_{\mathrm{abs}}=
 \sum_{s=1}^{S}\sum_{(u\stackrel{M}{\longrightarrow}v)\in\mathcal P_s}|M|,
\]
and the average amount of data sent per node is $V_{\mathrm{avg}}=\Vol_{\mathrm{abs}}/p$.

At the \emph{network} level, each transmission must be \emph{mapped} to the topology, i.e., implemented by routing the corresponding data from $u$ to $v$ along a physical path. The communicating nodes need not be adjacent in the topology. A message contributes its size to the load of every physical link on its route; different transmissions can contribute to the load of the same link.

Represent the physical network by a directed graph $G=(V,E)$, with the two directions of each connection represented by separate edges. On a torus, every vertex is a computing node. On a spine-leaf network, input vectors reside at hosts and paths may pass through forwarding switches. In the regular case with $p$ computing nodes of out-degree $D$, $|E|=pD$. All links have the same bandwidth, and $\beta$ is the time to transmit one bit along one link. We use a multiport model in which a node can use its incident links concurrently.

\subsection{The \texorpdfstring{$\alpha$--$\beta$}{alpha--beta} model and bandwidth deficiencies}
The classical cost model for an allreduce operation is the $\alpha$--$\beta$ model~\citep{rabenseifner2004optimization}. Its multiport form is
\[
 T_{\mathrm{allreduce}}=\alpha S+\beta\frac{V_{\mathrm{avg}}}{D},
\]
where $S$ is the number of rounds and $\alpha$ is the latency incurred by initiating a round. This expression treats the communication volume as evenly distributed over the available links.

To account for the mapping of transmissions to the topology, let $\ell_s(e)$ be the number of bits carried by directed physical link $e$ in round~$s$. We approximate the bandwidth cost of a round by its maximum link load and sum these costs over the rounds:
\begin{equation}\label{eq:roundcost}
 Q=\sum_{s=1}^{S}\max_{e\in E}\ell_s(e),
 \qquad T_{\mathrm{model}}=\alpha S+\beta Q.
\end{equation}
The following classical bound provides the reference communication volume~\citep{patarasuk2009bandwidth}.
\begin{proposition}\label{prop:minvol}
Let $G$ be a network with $p$ nodes. Assume that vectors of $N$ bits at each node need to be reduced. Then the average number of bits a node needs to send is at least
\begin{equation}\label{eq:minvol}
 V_{\min}=N\frac{2p-2}{p}.
\end{equation}
\end{proposition}
\begin{proof}
Take one position of the input vectors. Order the message deliveries within each round arbitrarily, keeping their contents fixed as at the beginning of that round. Let $v$ be the first node to obtain the complete result in this ordering. Each other node must have sent its contribution at least once, giving at least $p-1$ transmissions. After $v$ obtains the result, each of the other $p-1$ nodes still needs at least one incoming transmission to obtain it. Thus at least $2p-2$ transmissions are necessary for each position. In total, at least $2N(p-1)$ bits must be transmitted, or $2N(p-1)/p$ bits on average per node.
\end{proof}
A tree attains this bound with $p-1$ reduction transmissions and $p-1$ dissemination transmissions for each chunk.

A fruitful point of view on allreduce that this proposition suggests is to model it as a collection of graphs (typically, trees), where each graph records the sequence of communications used to reduce a given vector position. We distinguish three factors in the bandwidth cost of the resulting scheme.

\begin{definition}[Bandwidth deficiency]\label{def:psi}
The bandwidth deficiency is
\begin{equation}\label{eq:psi}
 \Psi=\frac{V_{\mathrm{avg}}}{V_{\min}}
      =\frac{\Vol_{\mathrm{abs}}}{2N(p-1)}\geq1.
\end{equation}
It measures the communication volume of the abstract scheme relative to the minimum.
\end{definition}

\begin{definition}[Multi-hop path deficiency]\label{def:theta}
Fix a mapping of the abstract allreduce scheme into the topology. For each transmission $(u\stackrel{M}{\longrightarrow}v)\in\mathcal P_s$, let $P_s(u,v;M)$ be the chosen physical path and let $\mathrm{len}(P_s(u,v;M))$ denote its length in physical edges. The multi-hop path deficiency is
\begin{equation}\label{eq:theta}
 \Theta=
 \frac{\displaystyle\sum_{s=1}^{S}\sum_{(u\stackrel{M}{\longrightarrow}v)\in\mathcal P_s}
 |M|\,\mathrm{len}(P_s(u,v;M))}
 {\displaystyle\sum_{s=1}^{S}\sum_{(u\stackrel{M}{\longrightarrow}v)\in\mathcal P_s}|M|}
 =\frac{\Vol_{\mathrm{phys}}}{\Vol_{\mathrm{abs}}}\geq1,
\end{equation}
where $\Vol_{\mathrm{phys}}=\sum_s\sum_e\ell_s(e)$ counts the volume over physical links. The two expressions are equivalent because each message contributes its size to every link on its route. Thus $\Theta$ is the message-size-weighted average path length, and $\Theta=1$ exactly when every transmission uses a path of length one, between adjacent nodes.
\end{definition}

\Needspace{9\baselineskip}
\begin{definition}[Load-balancing deficiency]\label{def:omega}
Define the average link load in round $s$ as $\bar\ell_s=|E|^{-1}\sum_e\ell_s(e)$. The load-balancing deficiency is
\begin{equation}\label{eq:omega}
 \Omega=\frac{\sum_s\max_e\ell_s(e)}{\sum_s\bar\ell_s}\geq1.
\end{equation}
It measures the multiplicative increase in the modeled bandwidth cost caused by per-round link-load imbalance. In particular, equal total traffic over the full execution does not imply equal traffic in each round.
\end{definition}

For a regular topology with $|E|=pD$, substituting these definitions into Equation~\eqref{eq:roundcost} gives
\begin{equation}\label{eq:factorized}
 T_{\mathrm{model}}=\alpha S+
 \beta\frac{2N(p-1)}{pD}\Psi\Theta\Omega.
\end{equation}
Approximating $(p-1)/p$ by one yields the form $\alpha S+\beta(2N/D)\Psi\Theta\Omega$. We retain the finite-$p$ factor when deriving bounds.

\begin{definition}[Full bandwidth optimality]
An allreduce scheme is \emph{fully bandwidth-optimal} if $\Psi=\Theta=\Omega=1$: it transmits the minimum volume, incurs no additional traffic from multi-hop routing, and balances link loads perfectly in every round.
\end{definition}
For schemes with these properties, the bandwidth term is fixed. Minimizing their $\alpha$--$\beta$ execution time therefore reduces to minimizing the number of rounds.

\subsection{Torus notation}
The equal-sided torus $\Tor_n^d$ has vertex set $\mathbb Z_n^d$, with neighbors differing by $+1$ or $-1$ in one coordinate. For $n\geq3$, its degree is $D=2d$ and its diameter is $H=d\lfloor n/2\rfloor$. For the two-dimensional rectangular torus $\Tor_{m,n}$, the diameter is $H=\lfloor m/2\rfloor+\lfloor n/2\rfloor$.

\subsection{Related work}\label{sec:related}
We briefly review representative allreduce baselines for torus and torus-like interconnects, comparing their communication rounds, transmitted volume, physical routes, and link utilization. Appendix~\ref{app:baselinedeficiencies} gives the quantitative comparison in our deficiency model and its derivations.

\paragraph{Hamiltonian ring allreduce.}
Ring allreduce~\citep{patarasuk2009bandwidth} performs reduce-scatter followed by allgather along Hamiltonian cycles in $2(p-1)$ rounds. When the cycles use physical links, it achieves $\Psi=\Theta=1$. A balanced multiport implementation runs several rings on separate chunks and distributes traffic equally across physical links in every round, giving $\Omega=1$. Such an implementation is \emph{fully bandwidth-optimal}. Its main drawback is the linear number of rounds.

\paragraph{Bucket algorithm.}
Bucket~\citep{jain2010optimal,sack2015collective} performs reduce-scatter and allgather along the coordinate rings of a torus, using only nearest-neighbor exchanges. On $\Tor_n^d$, its balanced multiport schedule is \emph{fully bandwidth-optimal}, with $\Psi=\Theta=\Omega=1$, and requires $2d(n-1)$ rounds. SP-allreduce preserves these bandwidth properties while approximately halving the number of rounds.

\paragraph{Swing.}
Swing~\citep{de2024swing} refines doubling-like allreduce algorithms by choosing partners with shorter physical routes on tori, thereby reducing the multi-hop path deficiency $\Theta$. For $n=2^m$, SwingL requires $\log_2 p$ rounds, while SwingB performs reduce-scatter and allgather in $2\log_2 p$ rounds. SwingB achieves minimum abstract communication volume, with $\Psi=1$, whereas SwingL incurs additional abstract traffic.

\paragraph{Trivance.}
Trivance~\citep{juerss2026trivance} extends logarithmic-round allreduce schemes by using simultaneous exchanges with two peers to combine contributions from three groups per step. For $n=3^t$, TrivanceL requires $\log_3 p$ rounds, while TrivanceB performs reduce-scatter and allgather in $2\log_3 p$ rounds. TrivanceB achieves minimum abstract communication volume, with $\Psi=1$, whereas TrivanceL incurs additional abstract traffic.

\paragraph{Other approaches.}
A complementary line of work aims to \emph{synthesize} collective schedules automatically for a given accelerator fabric, rather than relying on a fixed closed-form algorithm. Most such systems generate \emph{step schedules} that explicitly track data dependencies across synchronized rounds and optimize link usage within each round~\citep{cai2021synthesizing,shah2023taccl,cao2025syccl}. To keep the search tractable, these approaches use communication sketches in TACCL or symmetry-driven reductions in SyCCL~\citep{shah2023taccl,cao2025syccl}. Greedy synthesizers such as TACOS and MultiTree, as well as tree-based constructions such as Blink, provide further approaches to topology-aware collective design~\citep{won2024tacos,huang2021communication,wang2020blink}.

ForestColl~\citep{zhao2024forestcoll} constructs tree-flow schedules via spanning-tree packing, aiming to maximize throughput by balancing aggregate link load over the entire collective rather than enforcing perfectly balanced utilization in each synchronized round. Our analysis instead tracks the round structure and per-round physical-link loads.

Finally, a separate line of work co-designs the interconnect itself to improve collective performance~\citep{zhao2025efficient}. This direction is orthogonal to our setting, where the topology is fixed. SP-allreduce uses an explicit construction with formal bandwidth and round-count guarantees on equal-sided tori, and adapts the same tree-symmetrization approach to other topologies.

\section{SP-allreduce: Construction and Symmetry}\label{sec:symmetry}
\subsection{SP-allreduce construction}
We first construct one tree and then propagate it using torus symmetries to balance the load.

\begin{description}
\item[Step 1 (Data partition).]
Each node splits its data into $2d\cdot n^d$ equal chunks, indexed by triples
\[
 (x,\varepsilon,j)\in\mathbb Z_n^d\times\{+1,-1\}\times\{0,\ldots,d-1\}.
\]
Here $x$ specifies the root at which the chunk will be reduced, while $\varepsilon$ and $j$ specify the reflection and coordinate permutation used to construct its tree.

\item[Step 2 (Construction of allreduce trees).]
Fix the root $\boldsymbol 0=(0,\ldots,0)$ and construct any BFS (breadth-first search) tree $T_0$ of the torus rooted at $\boldsymbol 0$. Figure~\ref{fig:basetrees}(a) shows an example on $8\times8$. Orient all edges toward the root. Construct the other trees as follows.
\begin{description}
\item[Reflection.]
Set $T_0^+=T_0$ and obtain $T_0^-$ by reflecting every coordinate. Replace each vertex $(u_1,\ldots,u_d)$ by $(-u_1,\ldots,-u_d)$, with coordinates taken modulo $n$. Thus an edge $u\to v$ becomes $-u\to -v$.
\item[Cyclic coordinate shifts.]
Let $\rho(u_1,\ldots,u_d)=(u_2,\ldots,u_d,u_1)$. Apply the $d$ transformations $\rho^j$, $j=0,\ldots,d-1$, to each of $T_0^+$ and $T_0^-$, obtaining the trees $T_0^{\varepsilon,j}$.
\item[Translation to all roots.]
For each node $x\in\mathbb Z_n^d$, translate each $T_0^{\varepsilon,j}$ by $x$. An edge $u\to v$ becomes $u+x\to v+x$, with all coordinates taken modulo $n$. Denote the translated tree by $T_x^{\varepsilon,j}$.
\end{description}
As a result, for each root $x$ we obtain $2d$ aggregation trees spanning the entire torus. The chunk indexed by $(x,\varepsilon,j)$ is assigned to $T_x^{\varepsilon,j}$.

\item[Step 3 (Reduce-scatter phase).]
The depth of each tree equals the torus diameter $H=d\lfloor n/2\rfloor$. Communication proceeds in $H$ rounds. For each chunk $(x,\varepsilon,j)$, in round $i=1,\ldots,H$, every vertex at tree depth $H-i+1$ sends its partial aggregate to its parent. The parent combines the received values with its own contribution. At the end of this phase, root $x$ holds the reduced chunk $(x,\varepsilon,j)$ from all nodes.

\Needspace{6\baselineskip}
\item[Step 4 (Allgather phase).]
The allgather phase reverses the reduce-scatter: for each index $(x,\varepsilon,j)$, root $x$ sends the reduced chunk back to all nodes along the same tree in the opposite direction. This takes another $H$ rounds. Every node then holds all reduced chunks and thus the allreduce result.
\end{description}

Communication along all trees happens in parallel. In round $i$, along each directed edge $v\to w$, vertex $v$ sends together the chunks assigned to trees that use this edge in that round.

\begin{figure}[tbp]
 \centering
 \includegraphics[width=\linewidth]{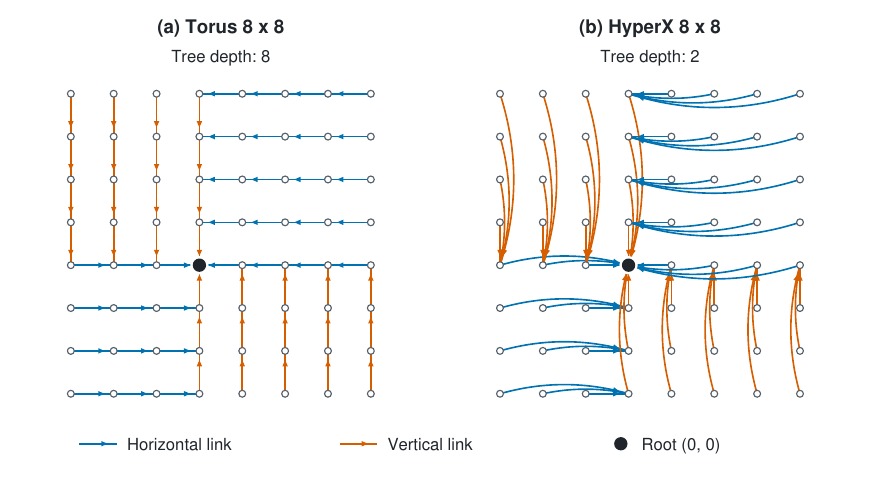}
 \Description{Base trees on an 8 by 8 torus and HyperX. The torus tree uses adjacent grid nodes and has depth eight. The HyperX windmill connects the root to its row and column and all other nodes to this cross, giving depth two. Arrows point toward the root.}
 \caption{Examples of base trees on $8\times8$: (a) a BFS tree on a torus; (b) the HyperX windmill used in Section~\ref{sec:hyperxconstruction}. Only tree edges are shown, with arrows in the reduction direction. Each curved edge in (b) is a single physical link between nodes in the same row or column.}
 \label{fig:basetrees}
\end{figure}

In algebraic terms, Step 2 applies the subgroup
\[
 \Gamma=\{u\mapsto x+\varepsilon\rho^j(u):
 x\in\mathbb Z_n^d,\ \varepsilon\in\{+1,-1\},\ 0\leq j<d\}
\]
of the torus automorphism group. It has $2dn^d$ elements for $n\geq3$. Translations move a directed edge to any position, cyclic coordinate shifts move it to any axis, and reflection reverses its direction. These transformations can map any directed physical link to any other. We next express this construction in general terms and prove that it balances the load in every round.

\subsection{A communication tree and its symmetric copies}
A base construction consists of a rooted spanning tree $T$ over the nodes, a physical path $R_{uv}$ for every tree edge, and a round assignment $\sigma$. The tree is logical: its adjacent vertices need not be adjacent in $G$. In reduce-scatter, edges point toward the root; in allgather, the aggregate follows the same tree and physical paths in reverse.

For a tree of edge depth $h$, reduce-scatter takes $h$ rounds. In round $s$, vertices at depth $h-s+1$ send to their parents. Thus each vertex receives and reduces the contributions from its children before sending its aggregate to its parent. Reversing these rounds gives $h$ allgather rounds.

Let $\Gamma\leq\operatorname{Aut}(G)$ be a subgroup of the automorphism group of $G$. For each $\gamma\in\Gamma$, apply $\gamma$ to the tree and its physical routes, retaining the round assigned to each transmission. In particular, a route $(v_0,v_1,\ldots,v_m)$ becomes $(\gamma(v_0),\gamma(v_1),\ldots,\gamma(v_m))$. The resulting construction is denoted by $(\gamma T,\gamma R,\gamma\sigma)$.

Split each input vector into $|\Gamma|$ equal chunks, one for each $\gamma\in\Gamma$, and run the corresponding tree schedules in parallel. Chunks with the same sender, receiver, round, and physical route can be packed into a single message.

\subsection{Per-round balance}
The \emph{orbit} of a directed physical edge $e$ is the set $\{\gamma(e):\gamma\in\Gamma\}$. Thus two edges belong to the same orbit if a transformation in $\Gamma$ maps one to the other. The action is \emph{directed-edge transitive} if all directed physical edges belong to a single orbit.

\begin{lemma}[Orbit balance]\label{lem:balance}
For the equal-chunk family indexed by $\Gamma$, links in the same directed-edge orbit carry equal volume in each round. If $\Gamma$ is transitive on all directed physical links, then $\Omega=1$.
\end{lemma}
\begin{proof}
Let $a_s(e)$ be the base load on physical link $e$ in round $s$, for a chunk of unit size. With chunk size $c=N/|\Gamma|$, the transformed family has
\[
 \ell_s(e)=c\sum_{\gamma\in\Gamma}a_s(\gamma^{-1}e).
\]
For any $\eta\in\Gamma$, replacing $\gamma$ by $\eta\gamma$ shows that $\ell_s(\eta e)=\ell_s(e)$. Thus loads are constant within each orbit in every round. If there is one orbit, the maximum equals the average in each round, which gives $\Omega=1$ by Equation~\eqref{eq:omega}.
\end{proof}

This argument applies to any network with such a group of automorphisms and to any base tree with specified routes and rounds. In particular, the base tree need not be a BFS tree of the topology.

\paragraph{Several edge orbits.}\label{sec:orbitweights}
If there is more than one edge orbit, symmetrization balances the load within each orbit, but the loads of different orbits can differ. It therefore does not by itself guarantee $\Omega=1$.

\section{Fully Bandwidth-Optimal Allreduce on Tori}\label{sec:torus}
\subsection{Bandwidth guarantees}
\begin{theorem}[Equal-sided torus schedule]\label{thm:torus}
On $\Tor_n^d$, $n\geq3$, SP-allreduce uses $2dn^d$ chunks and $2H=2d\lfloor n/2\rfloor$ rounds. It achieves $\Psi=\Theta=\Omega=1$, and its $\alpha$--$\beta$ running time is
\begin{equation}\label{eq:toruscost}
 T_{\mathrm{SP}}=2\alpha H+\beta\frac{2N(p-1)}{pD},
 \qquad p=n^d,\quad D=2d.
\end{equation}
\end{theorem}
\begin{proof}
For each chunk, the reduction and dissemination each traverse exactly the $p-1$ edges of its tree once. Summing chunk sizes gives $\Vol_{\mathrm{abs}}=2N(p-1)$ and hence $\Psi=1$. Every tree edge is one physical hop, so $\Theta=1$. The group is directed-edge transitive, so Lemma~\ref{lem:balance} gives $\Omega=1$. Breadth-first tree depth equals root eccentricity, which is $H$ at every vertex of the torus. The prescribed level schedule therefore uses $H$ rounds per phase. Substitution in Equation~\eqref{eq:factorized} proves the cost expression.
\end{proof}

\subsection{Optimality of SP-allreduce}\label{sec:optimality}
We analyze each position of the input vector separately. Fix a position $i$, which we call an elementary chunk in this argument. Its complete aggregate is the result of reducing the contributions at position $i$ from all $p$ nodes. For a fully bandwidth-optimal scheme, $\Psi=1$ implies exactly $2p-2$ transmissions for this chunk, and $\Theta=1$ implies that every transmission connects adjacent nodes.

Two structural facts apply to this chunk. In the first round that produces its complete aggregate, the result becomes available either at a single node or simultaneously at two adjacent nodes exchanging complementary partial aggregates (Lemma~\ref{lem:equality}, \emph{Structure in the equality case}). After that round, every transmission of this chunk carries its complete aggregate to a node that has not yet obtained it (Lemma~\ref{lem:dissemination}, \emph{Dissemination after the first complete aggregate}). Since allreduce completes only when every node has the reduced value at every vector position, a lower bound on the rounds required for this chunk also applies to the entire operation.

\begin{theorem}[Optimality of SP-allreduce]\label{thm:optimality}
On the equal-sided torus $\Tor_n^d$, $n\geq3$, let $S_{\min}$ be the minimum number of rounds among all fully bandwidth-optimal allreduce schemes. With $H=d\lfloor n/2\rfloor$,
\[
 S_{\min}=\begin{cases}2H,&n\text{ odd},\\2H-1,&n\text{ even}.\end{cases}
\]
SP-allreduce uses $2H$ rounds. Its $\alpha$--$\beta$ running time is therefore minimum in this class for odd $n$ and exceeds the minimum by $\alpha$ for even $n$.
\end{theorem}
\begin{proof}
Fix an arbitrary vector position $i$ and let $r_*$ be the first round after which at least one node holds its complete aggregate. Throughout the proof, all inputs, aggregates, and transmissions refer to this elementary chunk. Each node of $\Tor_n^d$ has eccentricity $H$. For two adjacent nodes, their joint eccentricity---the maximum distance from any node to the nearer of the two---is
\[
 \rho=\begin{cases}H,&n\text{ odd},\\H-1,&n\text{ even}.\end{cases}
\]
These distance identities are proved in Corollary~\ref{cor:torusrounds}.

If a single node first obtains the result, every input must reach it by round $r_*$, so $r_*\geq H$. By Lemma~\ref{lem:dissemination}, at least $H$ further rounds are necessary to deliver the result to all nodes. Hence $S\geq2H$.

If two adjacent nodes first obtain the result, Lemma~\ref{lem:equality} shows that they exchange complementary aggregates in round $r_*$. Every input must therefore reach one of these nodes before the exchange, giving $r_*\geq\rho+1$. Dissemination requires at least $\rho$ further rounds, so $S\geq2\rho+1$. Combining the two cases gives $S\geq2H$ for odd $n$ and $S\geq2H-1$ for even $n$.

Theorem~\ref{thm:torus} establishes that SP-allreduce is fully bandwidth-optimal and takes $2H$ rounds. It attains the bound for odd $n$. For even $n$, a BFS forest with two adjacent roots has depth $H-1$. Reducing each component to its root, exchanging the root aggregates, and disseminating the result along the forest take $2H-1$ rounds in total. Symmetrizing this schedule gives a fully bandwidth-optimal scheme; Proposition~\ref{prop:evensharpness} proves this construction. Thus SP-allreduce takes exactly one more round than the optimum for even $n$.

All fully bandwidth-optimal schemes have the same bandwidth term in Equation~\eqref{eq:factorized}. Their modeled times differ only through $\alpha S$, proving the running-time claims.
\end{proof}

For even $n$, SP-allreduce uses $dn$ rounds, whereas the conventional multiport Bucket count is $2d(n-1)$. The ratio is $2(n-1)/n$, approaching two as $n$ grows. Both schedules have the same bandwidth term in the model; the smaller number of rounds reduces the latency term $\alpha S$.

\section{\texorpdfstring{$k$-SP-allreduce}{k-SP-allreduce} on Tori}\label{sec:khop}
In SP-allreduce, every transmission connects physically adjacent nodes. We now allow nonadjacent nodes to communicate within a single round by routing their transmissions along shortest physical paths. This permits communication trees of smaller depth and therefore fewer rounds. The tree construction retains minimum abstract communication volume, with $\Psi=1$, and the same symmetrization on equal-sided tori gives $\Omega=1$. Using paths longer than one physical link introduces $\Theta>1$. The experiments show that this trade-off extends the range of message sizes for which the tree schedules outperform the compared algorithms.

A rooted spanning tree is a $k$-hop tree if every root-to-leaf path contains at most $k$ tree edges. A tree of actual depth $h(T)\leq k$ requires $h(T)$ reduce-scatter rounds followed by $h(T)$ allgather rounds. Here $k$ bounds the depth of the communication tree, not the physical length of an individual route.

Assign edge $(u,v)$ the weight $\dist_G(u,v)$, and let
\begin{equation}\label{eq:treeweight}
 W(T)=\sum_{uv\in E(T)}\dist_G(u,v).
\end{equation}
Each tree still has $p-1$ edges, so the symmetrized family transmits $2N(p-1)$ bits at the abstract level and $2NW(T)$ bits over physical links. Consequently,
\begin{equation}\label{eq:depththeta}
 \Psi=1,\qquad\Theta=\frac{W(T)}{p-1}.
\end{equation}
On an equal-sided torus, the subgroup from Section~\ref{sec:symmetry} gives $\Omega=1$ by Lemma~\ref{lem:balance}. Substitution into Equation~\eqref{eq:factorized} yields
\begin{equation}\label{eq:depthcost}
 T_{\mathrm{model}}(T)=2\alpha h(T)+\frac{2\beta N}{pD}W(T).
\end{equation}
For a fixed depth bound $k$, minimizing the bandwidth term therefore reduces to finding a minimum-weight rooted $k$-hop spanning tree on the complete graph of network nodes, with shortest-path distances as weights. For trees of the same depth, the latency term is also fixed. We compare different depths in the experiments. The general minimum-weight $k$-hop spanning-tree problem is NP-hard, already for $k=2$~\cite{dahl1998twohop}. We use an integer linear programming (ILP) formulation of the minimum-weight rooted spanning-tree problem with a prescribed depth bound. On $8\times8$, the solvers obtain minimum-weight trees for the tested depth bounds. For related formulations with hop constraints, see~\citep{gouveia1996multicommodity}.

\section{Extensions to Other Topologies}\label{sec:beyondconstruction}
We adapt the SP-allreduce approach to other topologies by choosing a base communication tree, shortest physical routes for its edges, and a subgroup of network automorphisms. We briefly describe the resulting constructions for rectangular tori, HyperX, and spine-leaf networks.

\subsection{Rectangular tori}\label{sec:rectconstruction}
For a two-dimensional rectangular torus $\Tor_{m,n}$, use translations and global reflection:
\[
 \Gamma=\{(x,y)\mapsto(a+\varepsilon x,b+\varepsilon y):
 (a,b)\in\mathbb Z_m\times\mathbb Z_n,\ \varepsilon\in\{+1,-1\}\}.
\]
This subgroup of $\operatorname{Aut}(\Tor_{m,n})$ has $2mn$ elements and two directed-edge orbits: horizontal and vertical links. One base tree need not distribute its traffic equally between these orbits. Two suitably chosen BFS trees suffice.

\begin{theorem}[Two-tree balance on rectangular tori]\label{thm:rectangular}
For every $m,n\geq3$, there exist two BFS trees rooted at $(0,0)$ whose families under $\Gamma$, with $4mn$ equal chunks in total, give a fully bandwidth-optimal allreduce schedule in $2H$ rounds, where $H=\lfloor m/2\rfloor+\lfloor n/2\rfloor$.
\end{theorem}

The proof is given in Appendix~\ref{app:rect}.

The experiments in Section~\ref{sec:recteval} use this two-tree construction on all three rectangular tori.

\subsection{HyperX}\label{sec:hyperxconstruction}
In a HyperX $n\times n$, each node is connected to the other $2(n-1)$ nodes of its row and column~\cite{ahn2009hyperx,domke2019hyperx}. Its graph is $K_n\mathbin{\square}K_n$ and has diameter two. The base tree has depth two. The root is connected to every other node in its row and column, forming a cross. Nodes in the upper-left and lower-right quadrants attach to the cross node in their column; nodes in the other two quadrants attach to the cross node in their row. This produces the windmill pattern illustrated in Figure~\ref{fig:basetrees}(b).

The subgroup consists of translations of $\mathbb Z_n^2$ and the four rotations by multiples of $90^\circ$. Its directed-edge orbits are indexed by $\min(\Delta,n-\Delta)$, where $\Delta$ is the difference in the changing coordinate modulo $n$.

At each reduction round, the windmill uses edges from every orbit in proportion to its size. For even $n$, the orbit with difference $n/2$ has half as many directed links and receives half as many tree edges. Symmetrization therefore gives equal load on all directed links in every round.

\subsection{Spine-leaf}
We consider a two-tier spine-leaf network based on the Clos architecture~\citep{clos1953study}, in which every leaf switch is connected to every spine switch. Let each leaf have $q$ hosts, let there be $q$ spines and $l$ leaves, and denote host $i$ on leaf $j$ by $h_{i,j}$. We use a shortest-path base tree on the hosts and symmetrize it using the following subgroup of topology automorphisms.

Host and leaf indices are represented by tuples. The shifts add fixed offsets componentwise, with each component taken modulo its number of possible values. The subgroup consists of the simultaneous shifts
\[
h_{i,j}\mapsto h_{i+a,\,j+b},\qquad
L_j\mapsto L_{j+b},\qquad
S_i\mapsto S_{i+a}.
\]
The same host-index shift is applied on every leaf and to the spine indices. Splitting the vector into $ql$ equal chunks and applying these transformations generates one routed tree rooted at each host.

\section{Experimental Evaluation}\label{sec:evaluation}
We evaluate SP-allreduce on square and rectangular tori and $k$-SP-allreduce on square tori. Both are compared with state-of-the-art allreduce algorithms for tori. HyperX and spine-leaf experiments demonstrate the extension of the method to other topologies, using the torus algorithms as references. We examine how closely the shortest-path schedules approach the bandwidth upper bound, how reducing the tree depth changes the range of message sizes at which they outperform the compared algorithms, and how this range shifts with topology shape, link bandwidth, and hop latency.

\subsection{Methodology}\label{sec:methods}
\paragraph{Simulator and setup.}
All runs use the Structural Simulation Toolkit (SST)~\cite{janssen2010simulator}, following the setup of~\cite{de2024swing}. The simulator models routers, links, and queues at packet level. Unless stated otherwise, links carry $400$\,Gb/s and a hop costs $400$\,ns ($100$\,ns of link latency and $300$\,ns of router processing); the reduction itself costs nothing. Each run is one allreduce: every node holds an input vector and starts at the same time. The simulator is deterministic.

The square-torus experiments cover $8\times8$, $16\times16$, $32\times32$, and $64\times64$. Figure~\ref{fig:square} shows absolute goodput on $8\times8$ and $64\times64$; Figure~\ref{fig:envelope}(a) and Table~\ref{tab:depthselection} report results for all four sizes. The rectangular tori have $1024$ nodes and shapes $64\times16$, $128\times8$, and $256\times4$. We also evaluate HyperX $32\times32$ and a spine-leaf fabric with $16$ spines, $32$ leaves, and $16$ hosts per leaf. On $8\times8$, the bandwidth sweep covers $100$, $200$, $400$, and $800$\,Gb/s, and $1.6$ and $3.2$\,Tb/s; the hop-latency sweep covers $200$, $400$, $800$, and $1600$\,ns.

\paragraph{Metric.}
Elements are $4$ bytes, message sizes are powers of two, and the reported time $T$ is the allreduce itself, measured at the last node to finish. For $p$ nodes, each holding an input vector of $N$ bits, we define goodput as
\begin{equation}\label{eq:goodput}
 G=\frac{N(p-1)}{pT}.
\end{equation}
Message sizes are stated in bytes, and goodput in Gb/s or Tb/s.

\paragraph{Goodput upper bound.}
Goodput is bounded above by half the node's aggregate outgoing-link capacity. By Proposition~\ref{prop:minvol}, allreduce sends at least $2N(p-1)/p$ bits per node on average, approximately twice the input-vector size. If each node's outgoing links have total capacity $B_v$, then
\[
 T\geq\frac{2N(p-1)}{pB_v},\qquad G\leq G_{\max}=\frac{B_v}{2}.
\]
On a $d$-dimensional torus, $B_v=2dB$ and $G_{\max}=dB$; on HyperX $n\times n$, $B_v=2(n-1)B$ and $G_{\max}=(n-1)B$. A spine-leaf host has one link of capacity $B$, giving $G_{\max}=B/2$. These bounds give the horizontal upper-bound lines in the figures.

\paragraph{Chunks.}
SP-allreduce splits the vector into $2dn^d$ logical chunks on $\Tor_n^d$. Chunks are scheduling units, not packets: transmissions with the same sender, receiver, route, and round are packed into one message. For the two tori shown in Figure~\ref{fig:square}, the smallest input that gives one four-byte element per chunk is $1$\,KiB on $8\times8$ and $64$\,KiB on $64\times64$.

\paragraph{Baselines.}
On tori, we compare with Bucket~\cite{jain2010optimal,sack2015collective}, Swing~\cite{de2024swing}, and Trivance~\cite{juerss2026trivance}. Ring~\cite{patarasuk2009bandwidth} is also included on rectangular tori. On HyperX and spine-leaf, Bucket, Swing, and Trivance serve as reference algorithms; the spine-leaf panel additionally includes recursive doubling~\cite{rabenseifner2004optimization} and Ring. These two experiments demonstrate the method's applicability beyond tori and its bandwidth utilization. The Swing and Trivance curves use the better of their latency-oriented and bandwidth-oriented variants at each message size.

The $k$-SP-allreduce curve reports the best measured goodput among the tested depths at each message size. Unless stated otherwise, gains are relative to the best compared algorithm at the same message size.

\subsection{SP-allreduce on square tori}\label{sec:envelopeeval}

At each message size, we select the highest measured goodput among the tested depth bounds for $k$-SP-allreduce and compare it with the highest goodput among Swing, Bucket, and Trivance. Figure~\ref{fig:envelope}(a) reports this relative improvement as a percentage. For the tested depth set $\mathcal K$, the percentage gain is
\begin{equation}\label{eq:gain}
 \gain(M)=100\left(
 \frac{\max_{k\in\mathcal K}G_k(M)}
 {\max\{G_{\mathrm{Swing}}(M),G_{\mathrm{Bucket}}(M),G_{\mathrm{Trivance}}(M)\}}-1\right),
\end{equation}
where $M$ is the message size in bytes.

At the default network setting, $k$-SP-allreduce improves goodput over the best of Swing, Bucket, and Trivance by up to $1.19\times$, attained on $16\times16$ at $1$\,MiB. The maximum gains on $8\times8$, $16\times16$, $32\times32$, and $64\times64$ are $8.0\%$, $19.0\%$, $17.7\%$, and $18.3\%$, attained at $1$, $1$, $4$, and $64$\,MiB, respectively. On $16\times16$, the gain exceeds $5\%$ from $128$\,KiB through $32$\,MiB; on $32\times32$, it exceeds $5\%$ from $1$ through $64$\,MiB.

The tested depth bounds, including the BFS reference, are $\mathcal K=\{1,\ldots,8\}$ on $8\times8$, $\{2,\ldots,16\}$ on $16\times16$, $\{2,\ldots,32\}$ on $32\times32$, and $\{1,\ldots,64\}$ on $64\times64$. Table~\ref{tab:depthselection} reports the depth bound attaining the minimum measured completion time at each message size under the default network setting.


\begin{table}[h]
 \centering
 \caption{Depth bound $k$ selected by the measured $k$-SP-allreduce envelope on square tori. Bound $n$ denotes the BFS reference on $n\times n$.}
 \label{tab:depthselection}
 \begin{tabular}{lrrrr}
 \toprule
 Input size & $8\times8$ & $16\times16$ & $32\times32$ & $64\times64$\\
 \midrule
 $64$\,KiB & $1$ & $3$ & $6$ & $10$\\
 $128$\,KiB & $5$ & $3$ & $6$ & $10$\\
 $256$\,KiB & $3$ & $3$ & $6$ & $10$\\
 $512$\,KiB & $5$ & $3$ & $6$ & $10$\\
 $1$\,MiB & $5$ & $3$ & $6$ & $10$\\
 $2$\,MiB & $5$ & $4$ & $6$ & $10$\\
 $4$\,MiB & $5$ & $7$ & $6$ & $10$\\
 $8$\,MiB & $8$ & $7$ & $9$ & $10$\\
 $16$\,MiB & $8$ & $7$ & $11$ & $10$\\
 $32$\,MiB & $8$ & $15$ & $12$ & $11$\\
 $64$\,MiB & $8$ & $16$ & $20$ & $14$\\
 $128$\,MiB & $8$ & $16$ & $22$ & $14$\\
 $256$\,MiB & $8$ & $16$ & $22$ & $14$\\
 \bottomrule
 \end{tabular}
\end{table}

Note that Bucket uses approximately twice as many rounds as SP-allreduce. Thus, an approximately $2\times$ goodput improvement is expected in the small-message regime. At $64$\,KiB, the measured goodput improvements relative to Bucket are $1.74\times$ on $8\times8$ and $2.00\times$ on $64\times64$ (Figure~\ref{fig:square}).

\subsection{Link-bandwidth sensitivity}
Increasing link bandwidth reduces the time spent transmitting a fixed volume, so a larger message is needed before bandwidth savings outweigh communication latency.

Figure~\ref{fig:envelope}(b) compares $k$-SP-allreduce with the best of Swing, Bucket, and Trivance at each message size. Positive goodput gains persist across the tested link rates, with the largest gain reaching $20.6\%$ at $3.2$\,Tb/s and $1$\,MiB.

\subsection{Hop-latency sensitivity}
Increasing physical hop latency makes communication latency more costly, so larger messages are needed for bandwidth savings to compensate for it.

Figure~\ref{fig:envelope}(c) compares $k$-SP-allreduce with the best of Swing, Bucket, and Trivance at each message size. At $1600$\,ns per hop, $k$-SP-allreduce reaches a gain of $20.2\%$ at $512$\,KiB. These measurements concern physical hop latency, which is included in the simulator but not separately modeled by Equation~\eqref{eq:factorized}.

\begin{figure}[tbp]
 \centering
 \includegraphics[width=\linewidth]{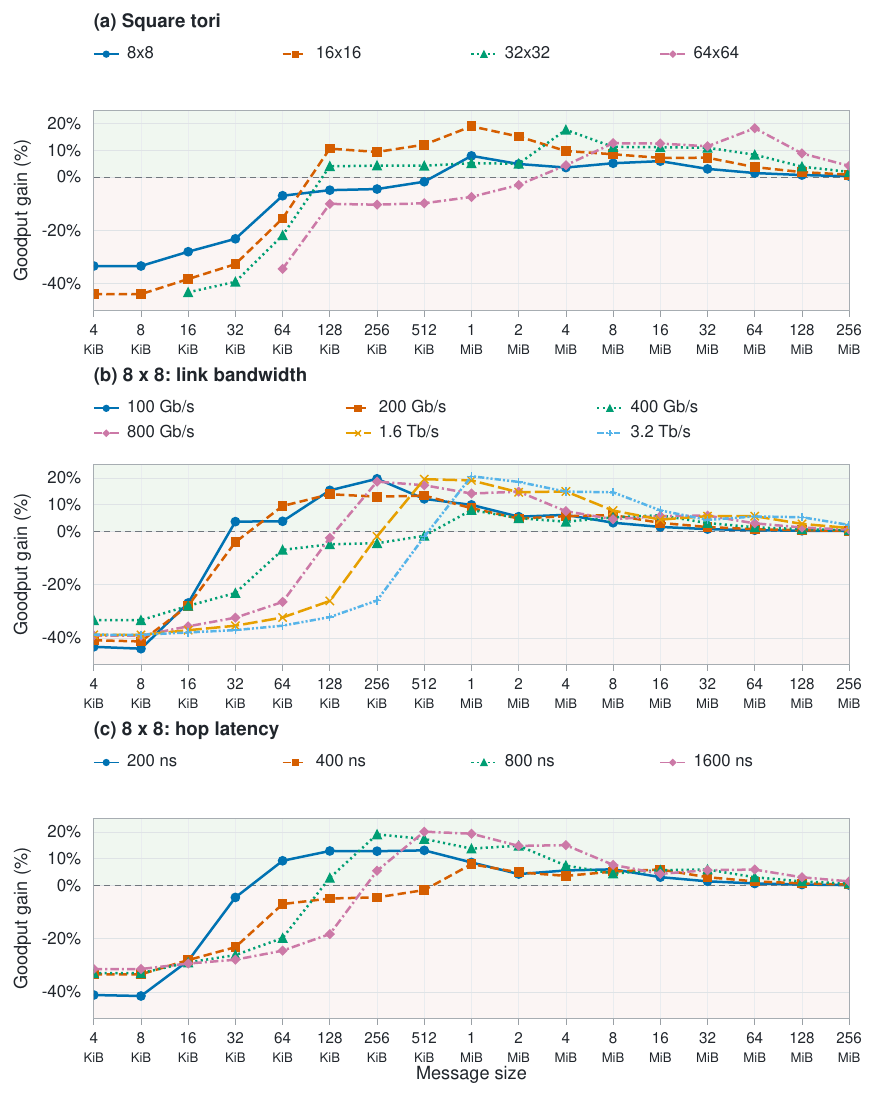}
 \Description{Measured k-SP-allreduce goodput gains on square tori and in bandwidth and hop-latency sweeps.}
 \caption{$k$-SP-allreduce goodput gain over the best of Swing, Bucket, and Trivance: (a) square tori, (b) link-bandwidth sweep on $8\times8$, and (c) hop-latency sweep on $8\times8$.}
 \label{fig:envelope}
\end{figure}

\begin{figure}[tbp]
 \centering
 \includegraphics[width=\linewidth]{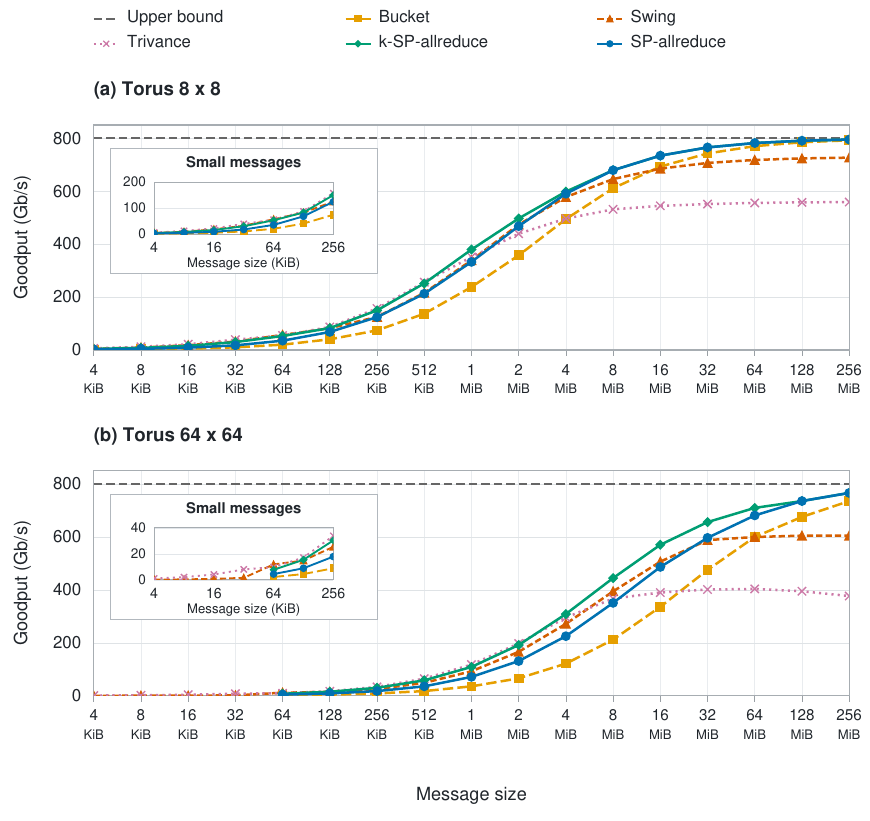}
 \Description{Goodput versus message size on 8 by 8 and 64 by 64 square tori, with insets for small messages.}
 \caption{Goodput on (a) $8\times8$ and (b) $64\times64$ tori at the default link bandwidth and hop latency. Insets show the small-message range on an expanded vertical scale.}
 \label{fig:square}
\end{figure}

\begin{figure}[tbp]
 \centering
 \includegraphics[width=0.985\linewidth]{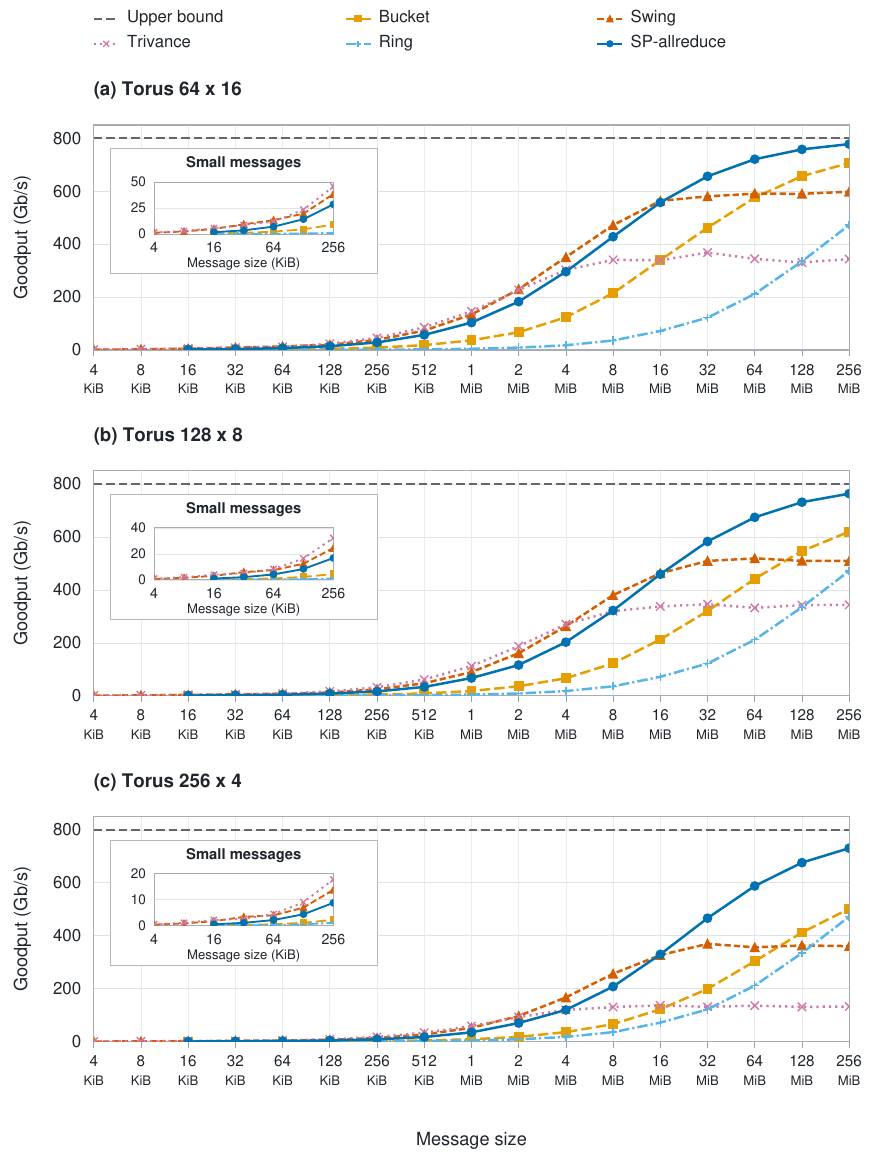}
 \Description{Goodput versus message size on three rectangular tori with 1024 nodes.}
 \caption{Goodput on rectangular tori with $1024$ nodes at the default network setting. SP-allreduce uses the two-tree construction on all three topologies. Insets show the small-message range on an expanded vertical scale.}
 \label{fig:rectangular}
\end{figure}

\subsection{Rectangular tori}\label{sec:recteval}
Figure~\ref{fig:rectangular} compares three $1024$-node tori: $64\times16$, $128\times8$, and $256\times4$. The corresponding SP-allreduce round counts are $80$, $136$, and $260$. All three have the same goodput bound of $800$\,Gb/s at the default link rate. All three experiments use the two-BFS-tree construction of Theorem~\ref{thm:rectangular}, with $\Psi=\Theta=\Omega=1$.

At $256$\,MiB, SP-allreduce reaches approximately $778$, $764$, and $731$\,Gb/s, or $97.2\%$, $95.5\%$, and $91.4\%$ of the common upper bound. The shortest-path construction and equal per-round loads across the horizontal and vertical link classes preserve high bandwidth utilization even on the most elongated torus. The measured values approach the same $800$\,Gb/s reference, rather than a smaller bound imposed by the topology's aspect ratio.

The relative advantage grows with elongation. The maximum SP-allreduce improvement over the best compared algorithm is approximately $1.22\times$ on $64\times16$ at $64$\,MiB, $1.34\times$ on $128\times8$ at $128$\,MiB, and $1.65\times$ on $256\times4$ at $64$\,MiB. At $256$\,MiB, Bucket reaches approximately $706$, $620$, and $502$\,Gb/s, while Swing reaches $598$, $509$, and $361$\,Gb/s. Thus increasing the aspect ratio lowers SP-allreduce goodput at this fixed message size, but lowers the competitors' goodput more strongly. The node count, degree, and bandwidth upper bound remain unchanged, so topology shape is a substantial factor in the comparison.

\subsection{HyperX and spine-leaf}\label{sec:othereval}

\begin{figure}[tbp]
 \centering
 \includegraphics[width=0.94\linewidth]{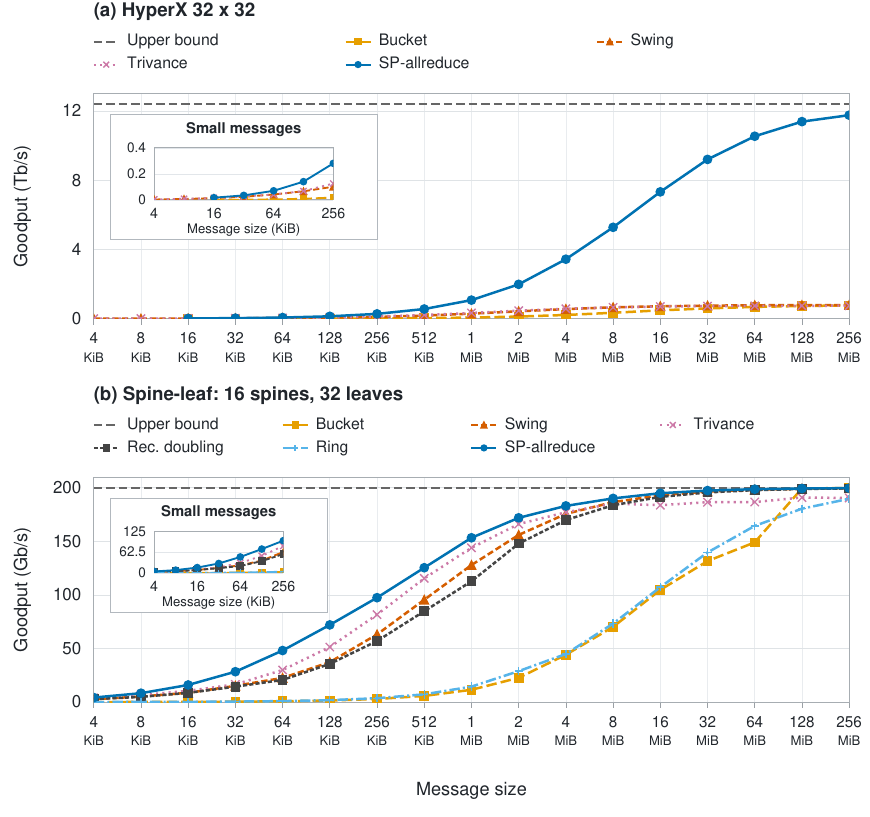}
 \Description{Goodput versus message size on HyperX 32 by 32 and the spine-leaf network with 16 spines and 32 leaves.}
 \caption{SP-allreduce beyond tori: goodput on (a) HyperX $32\times32$ and (b) spine-leaf with $16$ spines, $32$ leaves, and $16$ hosts per leaf. Comparisons use Swing, Trivance, and Bucket, with Ring and recursive doubling also shown on spine-leaf. HyperX values use Tb/s; spine-leaf values use Gb/s. Insets expand the small-message range.}
 \label{fig:beyond}
\end{figure}

We follow Swing~\cite{de2024swing} by conducting experiments on HyperX. We also run Trivance on HyperX using the same approach as for Swing. Figure~\ref{fig:beyond} shows the results on HyperX $32\times32$ and on a spine-leaf fabric with $16$ spines and $32$ leaves. The evaluated baselines are Swing, Trivance, and Bucket, with Ring and recursive doubling also included on spine-leaf.

On HyperX, the $62$ links per node give a goodput upper bound of $12.4$\,Tb/s at the default link bandwidth. SP-allreduce reaches approximately $11.76$\,Tb/s at $256$\,MiB, or $94.9\%$ of this bound. Swing, the fastest compared implementation at this size, reaches $0.794$\,Tb/s, or $6.4\%$ of the same bound. The measured goodput ratio between SP-allreduce and Swing is $14.82\times$.

A limitation of the evaluated Swing, Trivance, and Bucket implementations is that they are poorly suited to exploiting HyperX's connectivity: they use only $4$ of the $62$ links per node in each round, whereas SP-allreduce uses all $62$ links. This difference explains the large goodput gap. In our experiments, Swing reaches approximately $0.794$\,Tb/s, close to its four-link goodput limit of $4B/2=0.800$\,Tb/s at the default link bandwidth $B$.

On spine-leaf, each host has one $400$\,Gb/s link, giving a $200$\,Gb/s goodput upper bound. SP-allreduce reaches approximately $199.7$\,Gb/s at $256$\,MiB, or $99.8\%$ of the bound. At large message sizes, Bucket, Swing, and recursive doubling also approach this limit. SP-allreduce's maximum improvement over the best compared baseline is $1.72\times$ at $32$\,KiB, with approximately $28.4$\,Gb/s versus $16.5$\,Gb/s for Trivance.

\section{Conclusion}\label{sec:conclusion}
On equal-sided tori, symmetrizing a BFS tree under a group acting transitively on directed physical edges yields minimum abstract communication volume, one-hop transmissions, and perfect per-round load balance. The schedule takes $2H=2d\lfloor n/2\rfloor$ rounds. The minimum-transmission argument gives a sharp lower bound of $2H$ for odd side lengths and $2H-1$ for even side lengths. The described construction is optimal in the first case and uses exactly one additional round in the second.

Experimentally, the resulting algorithms improve over the best compared state-of-the-art algorithms for tori by up to $1.19\times$ on square tori and $1.65\times$ on rectangular tori across the tested configurations. On HyperX and spine-leaf, the approach nearly attains the bandwidth upper bound and improves over the best evaluated baselines by up to $14.82\times$ and $1.72\times$, respectively. The analytical guarantees use the homogeneous multiport model of Section~\ref{sec:model}; the simulations also account for physical hop latency.

\label{page:lastmain}
\clearpage
\bibliographystyle{ACM-Reference-Format}
\begingroup
\interlinepenalty=10000
\bibliography{refs}
\endgroup
\clearpage
\appendix

\section{Minimum Transmissions and the Structure of an Allreduce Schedule}\label{app:volume}
\subsection{A one-chunk information-flow model}

Fix one elementary chunk, corresponding to a single position \(i\) of the input vector. Let \(G\) be the communication graph and let
\(V=V(G)\), \(|V|=p\). For a vertex \(v\), let
\[
K_t(v)\subseteq V
\]
denote the set of original inputs at position \(i\) whose aggregate is known at \(v\) after
round \(t\). Initially \(K_0(v)=\{v\}\). In round \(t+1\), a transmission \(u\to v\) sends the aggregate over
\(K_t(u)\), where \(K_t(v)\) denotes the state at vertex \(v\) at the
beginning of the round. At the end of the round, \(v\) forms
\(K_{t+1}(v)\) by taking the union of \(K_t(v)\) and all sets received
in that round. The goal of allreduce for this chunk is to reach
\[
K_T(v)=V\qquad \text{for every }v\in V .
\]

For a state \(\mathcal K=\{K(v):v\in V\}\), define its covering number by
\[
c(\mathcal K):=\min\bigl\{r:\text{ there exist }v_1,\ldots,v_r
\text{ such that }K(v_1)\cup\cdots\cup K(v_r)=V\bigr\}.
\]
Thus initially \(c(\mathcal K_0)=p\), while if some vertex knows the full
aggregate then \(c=1\).

We shall use the following elementary observation. If one individual
transmission is added to a state, then the covering number can decrease
by at most one. Indeed, if after the transmission a cover of size \(c'\)
uses the newly formed set \(A\cup B\), then before the transmission one
could replace it by the two sets \(A\) and \(B\), obtaining a cover of
size at most \(c'+1\).

\subsection{Transmission lower bound}
\begin{proposition}[Transmission lower bound]\label{prop:transmissions}
In the above model, any allreduce schedule for one elementary chunk uses
at least \(2p-2\) transmissions.
\end{proposition}

\begin{proof}
Let \(\tau\) be the first round after which at least one vertex knows the
full aggregate \(V\). Put
\[
q:=c(\mathcal K_{\tau-1})
\]
and let
\[
r:=|\{v\in V:K_\tau(v)=V\}|
\]
be the number of vertices that first know the full aggregate after round
\(\tau\). Since no vertex is full before round \(\tau\), we have \(q\ge 2\).

Before round \(\tau\), the covering number has decreased from \(p\) to
\(q\). Since each individual transmission decreases it by at most one,
there must have been at least \(p-q\) transmissions before round \(\tau\).

Now consider round \(\tau\). Fix a vertex \(v\) with \(K_\tau(v)=V\). If
\(v\) receives \(s_v\) messages in this round, then the \(s_v\) received
sets together with the old set \(K_{\tau-1}(v)\) cover \(V\). Therefore
\(s_v+1\ge q\), and hence \(s_v\ge q-1\). Summing over the \(r\) vertices
that become full in round \(\tau\), the round contains at least
\(r(q-1)\) transmissions.

After round \(\tau\), there are \(p-r\) vertices which do not yet know
the full aggregate. Each of them must receive at least one later
transmission in order ever to become full. Hence there are at least
\(p-r\) transmissions after round \(\tau\).

Altogether, the total number of transmissions is at least
\[
(p-q)+r(q-1)+(p-r)
=2p-q+r(q-2)
=2p-2+(r-1)(q-2).
\]
Since \(r\ge1\) and \(q\ge2\), this is at least \(2p-2\).
\end{proof}

\subsection{Structure in the equality case}
\begin{lemma}[Structure in the equality case]\label{lem:equality}
Suppose an allreduce schedule for one elementary chunk uses exactly
\(2p-2\) transmissions. Let \(\tau\), \(q\), and \(r\) be as in the proof
above. Then, immediately before round \(\tau\), the inclusion-maximal
sets among
\[
\{K_{\tau-1}(v):v\in V\}
\]
are pairwise disjoint, form a partition of \(V\), and each maximal set has
a unique holder.

Moreover, if \(r>1\), then \(q=2\). In this case exactly two vertices can
become full in round \(\tau\), and this can only happen by a bidirectional
exchange of two complementary maximal aggregates. If transmissions are restricted to physical neighbors, the two holders are adjacent in the topology.
\end{lemma}

\begin{proof}
Since the schedule uses exactly \(2p-2\) transmissions, all inequalities
in the proof of the proposition must be tight. In particular, before
round \(\tau\) there are exactly \(p-q\) transmissions, and these
transmissions reduce the covering number from \(p\) to \(q\). Hence every
one of these transmissions reduces the covering number by exactly one.

We now prove by induction over the individual transmissions before round
\(\tau\) that the maximal aggregates are pairwise disjoint and have
unique holders. To avoid any ambiguity about simultaneous transmissions,
fix an arbitrary ordering of the transmissions inside each round, but
keep the payload of every transmission equal to the aggregate known by
the sender at the beginning of that round. Processing the transmissions
one by one with these fixed payloads gives the same state at the end of
the round, since all updates are unions.

Initially the maximal sets are the singletons \(\{v\}\), and each has a
unique holder. Suppose that before some individual transmission the
maximal sets are
\[
P_1,\ldots,P_h,
\]
pairwise disjoint, covering \(V\), and each \(P_i\) has a unique holder.
Then \(c=h\). Consider a transmission whose payload \(S\) is contained in
\(P_i\), and whose recipient currently holds a set contained in \(P_j\).

We claim that this transmission can decrease \(c\) only if
\[
S=P_i,\qquad K(v)=P_j,\qquad i\ne j,
\]
where \(v\) is the recipient. Indeed, if \(i=j\), the transmission only
combines information already contained in the same maximal component and
cannot reduce the number of components needed to cover \(V\). If
\(S\subsetneq P_i\), then the part \(P_i\setminus S\) still has to be
covered by the unique maximal aggregate \(P_i\); the new set does not
replace \(P_i\). Similarly, if the recipient holds a proper subset of
\(P_j\), then the old maximal aggregate \(P_j\) is still needed. Thus the
covering number cannot drop unless the whole component \(P_i\) is sent to
the unique holder of a different whole component \(P_j\).

In that case the transmission replaces the two maximal components
\(P_i\) and \(P_j\) by their union \(P_i\cup P_j\). The new maximal
components are
\[
P_i\cup P_j,\quad P_k\ (k\ne i,j),
\]
which are again pairwise disjoint and cover \(V\). The holder of
\(P_i\cup P_j\) is unique, namely the recipient. The old holder of
\(P_i\) may still possess \(P_i\), but \(P_i\) is now a proper subset of
the maximal set \(P_i\cup P_j\), so it is no longer a maximal aggregate.
Thus the induction continues.

Since every pre-\(\tau\) transmission must decrease the covering number,
the induction applies to all transmissions before round \(\tau\). Hence,
just before round \(\tau\), the maximal aggregates are pairwise disjoint,
form a partition of \(V\), and have unique holders.

It remains to analyze the first-full round. Exact optimality also gives
\[
2p-2=2p-2+(r-1)(q-2),
\]
and therefore
\[
(r-1)(q-2)=0.
\]
Thus either \(r=1\), or \(q=2\).

Assume \(r>1\). Then \(q=2\), so immediately before round \(\tau\) there
are exactly two maximal aggregates, say
\[
A\sqcup B=V,
\]
with unique holders \(a\) and \(b\). Each vertex that becomes full in
round \(\tau\) must receive at least \(q-1=1\) message, and by tightness
it receives exactly one relevant message. For one received aggregate
together with its old aggregate to cover \(V=A\sqcup B\), the vertex must
already hold one of \(A,B\) and receive the other. Since \(A\) and \(B\)
have unique holders, only \(a\) and \(b\) can become full in this round.
Thus \(r\le2\). If \(r=2\), then \(a\) receives \(B\) from \(b\), and
\(b\) receives \(A\) from \(a\), so the two holders exchange their
complementary aggregates; under one-hop routing, this exchange uses a physical edge. This proves the final
claim.
\end{proof}

\subsection{The number of rounds with one-hop transmissions}
For the remainder of the proof, $G$ is the physical topology and $\Theta=1$: every transmission follows a path of length one, so only adjacent network nodes communicate in a round.

\begin{definition}
For a set \(S\subseteq V(G)\), write
\[
\operatorname{ecc}(S):=\max_{x\in V(G)} \operatorname{dist}_G(x,S).
\]
In particular,
\[
\operatorname{ecc}(v):=\operatorname{ecc}(\{v\}),
\qquad
\operatorname{ecc}(uv):=\operatorname{ecc}(\{u,v\})
\]
for an edge \(uv\in E(G)\). \end{definition}

\begin{lemma}[Dissemination after the first complete aggregate]\label{lem:dissemination}
Consider the one-chunk model with one-hop transmissions. Suppose an allreduce
schedule uses exactly \(2p-2\) transmissions. Let \(\tau\) be the first
round after which some vertex knows the full aggregate \(V\), and let
\[
F_\tau:=\{v\in V:K_\tau(v)=V\}.
\]
Then every transmission after round \(\tau\) is sent by a vertex that is
already full at the beginning of that round, and its recipient becomes
full at the end of that round.

Consequently, every vertex \(y\) that is not full at time \(\tau\) receives
the full aggregate along a causal path starting from \(F_\tau\). In
particular, if the schedule finishes after \(T\) rounds, then
\[
\operatorname{dist}_G(y,F_\tau)\le T-\tau
\qquad\text{for every }y\in V.
\]
\end{lemma}

\begin{proof}
Let
\[
q:=c(\mathcal K_{\tau-1}),\qquad r:=|F_\tau|.
\]
From the transmission lower-bound proof, the number of transmissions is at
least
\[
(p-q)+r(q-1)+(p-r)
=
2p-2+(r-1)(q-2).
\]
Since the schedule uses exactly \(2p-2\) transmissions, all inequalities
in that proof are tight. In particular, after round \(\tau\) there are
exactly \(p-r\) transmissions.

Every vertex outside \(F_\tau\) must receive at least one transmission
after round \(\tau\), since otherwise it can never learn the full
aggregate. There are exactly \(p-r\) such vertices and exactly \(p-r\)
post-\(\tau\) transmissions. Hence each vertex outside \(F_\tau\) receives
exactly one post-\(\tau\) transmission, and no vertex that is already full
can receive a post-\(\tau\) transmission. Moreover, the unique
post-\(\tau\) transmission received by a vertex outside \(F_\tau\) must
make it full immediately.

It remains to show that the sender of such a transmission must already be
full. Let
\[
P_1,\ldots,P_q
\]
be the inclusion-maximal aggregates immediately before round \(\tau\).
By the equality-case structure, these sets are pairwise disjoint, form a
partition of \(V\), and each \(P_i\) has a unique holder.

Consider a post-\(\tau\) transmission \(u\to v\) whose recipient \(v\) is
not full at the beginning of the round. Since \(v\notin F_\tau\) and has
not yet received its unique post-\(\tau\) transmission, its current
aggregate is still the same as at time \(\tau\), and hence is contained in
one of the \(P_i\)'s. The same is true for \(u\) if \(u\) is not full at
the beginning of the round.

We claim that two non-full vertices cannot together cover \(V\). Indeed,
if \(r=1\), then one of the unique holders of the \(P_i\)'s became full in
round \(\tau\). If \(q=2\), the complementary maximal aggregate is still
held by its unique holder, but the other maximal aggregate is held only by
the full vertex; hence no two non-full vertices hold complementary
maximal aggregates. If \(q\ge3\), the union of two sets each contained in
single parts of the partition \(P_1,\ldots,P_q\) cannot cover \(V\).
If \(r>1\), then exact optimality gives \(q=2\), and the two unique holders
of the two maximal aggregates both become full in round \(\tau\). Thus any
non-full vertex holds a proper subset of one of the two parts, and again
two non-full vertices cannot cover \(V\).

Therefore, if the sender \(u\) were not full, then the union of the
payload sent by \(u\) and the old aggregate at \(v\) would still be a
proper subset of \(V\). This contradicts the fact that \(v\)'s unique
post-\(\tau\) transmission must make \(v\) full. Hence the sender is full
at the beginning of the round.

Thus fullness propagates after round \(\tau\) only by transmissions from
already full vertices to not-yet-full vertices. Since information can
traverse at most one edge per round, every vertex made full after
\(\tau\) lies on a causal path starting from \(F_\tau\), of length at most
the number of remaining rounds. Hence
\[
\operatorname{dist}_G(y,F_\tau)\le T-\tau
\]
for every \(y\in V\).
\end{proof}

\begin{corollary}[Vertex and edge eccentricities on equal-sided tori]\label{cor:torusrounds}
On $\Tor_n^d$, every node has eccentricity $H=d\lfloor n/2\rfloor$, and every adjacent pair has joint eccentricity $H-1$ for even $n$ and $H$ for odd $n$.
\end{corollary}
\begin{proof}
Distances in a Cartesian product add, and every node of $C_n$ has eccentricity $\lfloor n/2\rfloor$. Thus every node of $\Tor_n^d$ has eccentricity $H$. Consider adjacent nodes differing in coordinate $j$. In each of the other $d-1$ coordinates, the largest distance contribution is $\lfloor n/2\rfloor$. In coordinate $j$, the maximum distance to the nearer endpoint of an edge of $C_n$ is $n/2-1$ for even $n$ and $(n-1)/2$ for odd $n$. Adding these contributions gives $H-1$ and $H$, respectively.
\end{proof}

\subsection{Sharpness on even equal-sided tori}\label{app:sharpness}
\begin{proposition}[Attaining the even-side bound]\label{prop:evensharpness}
For even $n\geq4$, there is a fully bandwidth-optimal allreduce schedule on $\Tor_n^d$ taking $2H-1$ rounds, where $H=dn/2$.
\end{proposition}
\begin{proof}
Let $n$ be even and $H=d n/2$. Choose adjacent nodes $a,b$ of $\Tor_n^d$. Their joint eccentricity is $H-1$, as established in Corollary~\ref{cor:torusrounds}. A BFS forest rooted at $\{a,b\}$ assigns every other node a parent whose distance to this root set is one smaller. It has two components, $p-2$ edges, and depth $H-1$.

Reduce each component to its root in $H-1$ rounds, using $p-2$ transmissions. The roots then exchange their complementary aggregates in one round, using two transmissions. Both now know the complete result. Disseminating it along the forest takes another $H-1$ rounds and $p-2$ transmissions. The total is $2H-1$ rounds and $2p-2$ transmissions, all between physical neighbors.

Apply the torus subgroup from Section~\ref{sec:symmetry} to this entire schedule, including the forest, the central exchange, and all round assignments, and assign equal chunks to its copies. The proof of Lemma~\ref{lem:balance} depends only on the transformed round loads, so it also applies to this schedule. It gives $\Omega=1$, while the transmission count and one-hop routes give $\Psi=\Theta=1$. Thus the even-side lower bound is attained within the fully bandwidth-optimal class. This construction establishes sharpness of the bound; the SP-allreduce algorithm evaluated in the experiments remains the $2H$-round construction of Section~\ref{sec:symmetry}.

\end{proof}

\section{Two-Tree Balance on Rectangular Tori}\label{app:rect}
\begin{proof}[Proof of Theorem~\ref{thm:rectangular}]
Fix the root $(0,0)$ in $\Tor_{m,n}$, with $m,n\geq3$, and let $S_\ell$ be the set of nodes at distance $\ell$ from the root. Write $L_\ell=|S_\ell|$. A BFS tree chooses, for each node in $S_\ell$, a parent in $S_{\ell-1}$.

Let $A_\ell$ count the nodes of $S_\ell$ with $y=0$, and $B_\ell$ those with $x=0$. These nodes must choose a horizontal and a vertical parent edge, respectively. Each of the remaining $C_\ell=L_\ell-A_\ell-B_\ell$ nodes can choose either direction, by decreasing its distance in that coordinate.

We first show that
\begin{equation}\label{eq:rectquota}
 A_\ell,B_\ell\leq\lfloor L_\ell/2\rfloor.
\end{equation}
For $\ell=1$, both counts are two and $L_1=4$. For $\ell\geq2$, each axis contributes at most two nodes. If $A_\ell>0$, then $\ell\leq\lfloor m/2\rfloor$, and there are at least four mixed-coordinate nodes with coordinate distances $\ell-1$ and $1$. Indeed, the first distance is strictly below $m/2$ and has two choices, and distance one on the other cycle also has two choices. Therefore $L_\ell\geq5$ whenever $A_\ell>0$, which implies the first inequality. The same argument with the coordinates exchanged proves the second. If the relevant axis count is zero, its inequality is immediate.

In $T_1$, choose exactly $\lfloor L_\ell/2\rfloor$ horizontal parent edges from layer $S_\ell$; in $T_2$, choose exactly $\lceil L_\ell/2\rceil$. Both choices are feasible by Equation~\eqref{eq:rectquota}: the number of horizontal edges can be any integer from $A_\ell$ to $L_\ell-B_\ell$. The remaining edges are vertical. Choices can be made independently at each layer because every parent edge decreases distance to the root. They therefore define two spanning BFS trees, both of depth $H=\lfloor m/2\rfloor+\lfloor n/2\rfloor$.

Together, the two trees have exactly $L_\ell$ horizontal and $L_\ell$ vertical parent edges from layer $S_\ell$. Symmetrize both trees under the subgroup $\Gamma$ of translations and global reflection from Section~\ref{sec:rectconstruction}. There are $2mn$ transformations and two directed-edge orbits, each with $2mn$ edges. Divide the input into $4mn$ equal chunks, one per tree and transformation. In reduction round $s=H-\ell+1$, Lemma~\ref{lem:balance} makes the load constant within each orbit. The total traffic in each orbit is
\[
 \frac{N}{4mn}\,|\Gamma|L_\ell=\frac{N L_\ell}{2},
\]
so every directed link in either orbit carries $N L_\ell/(4mn)$ bits. The reversed allgather rounds have the same balance. Hence $\Omega=1$. Each tree uses exactly $p-1$ edges per phase, all between adjacent nodes, giving $\Psi=\Theta=1$. The schedule has $2H$ rounds, completing the proof.

\end{proof}

\section{Deficiencies of the Torus Algorithms}\label{app:baselinedeficiencies}
All volumes are in bits; $p=n^d$ and the torus has $2dp$ directed links. We assume equal-sized chunks and shortest coordinate routes. Swing is analyzed for $n=2^m$, $m\geq2$, and Trivance for $n=3^t$. These radix assumptions are specific to each algorithm and do not extend the formulas to arbitrary side lengths.

SwingB and SwingL split the input into $2d$ equal parts and run one subcollective on each, with pairs using opposite directions and rotated across coordinates. Trivance splits the input into $d$ equal-sized subvectors and runs one subcollective on each, using two communication directions. These are the multiport variants used in the experiments. Let $K_2=dm$ and $K_3=dt$ for their respective radix cases.

For Swing, define
\[
 a_k=\frac{2^{k+1}+(-1)^k}{3},\qquad
 A_m=\sum_{k=0}^{m-1}a_k=\frac{4n-3-(-1)^m}{6}.
\]
The message-size-weighted physical path lengths are $F_2(n,d)$ for SwingB with $n=2^m$ and $F_3(n,d)$ for TrivanceB with $n=3^t$:
\begin{align}
 F_2(n,d)&=\frac{1-2^{-d}}{1-n^{-d}}
       \sum_{k=0}^{m-1}a_k2^{-dk},\label{eq:baselinef2}\\
 F_3(n,d)&=\frac{1-3^{-d}}{1-n^{-d}}
       \sum_{k=0}^{t-1}3^{(1-d)k}.\label{eq:baselinef3}
\end{align}
The following derivations use the stated routing and multiport assumptions and count communication volumes and route lengths.

\subsection{SP-allreduce and Bucket}
SP-allreduce and the balanced multiport Bucket schedule~\citep{jain2010optimal,sack2015collective} both transmit the minimum abstract volume $2N(p-1)$. Every transmission connects physical neighbors, and all directed links have equal load in each round. Thus both have $\Psi=\Theta=\Omega=1$. Theorem~\ref{thm:torus} proves these properties for SP-allreduce.

\subsection{SwingB}
In reduce-scatter round $s\in\{0,\ldots,K_2-1\}$, each of the $2d$
subcollectives sends a message of size
\[
 b_s=\frac{N}{2d\,2^{s+1}}
\]
per node, at distance $a_{\lfloor s/d\rfloor}$. Each node therefore sends
$N/2^{s+1}$ bits in this round. The reverse allgather has the same volumes
and path lengths, giving
\[
 \Vol_{\rm abs}=2pN\sum_{s=0}^{K_2-1}2^{-s-1}
 =2pN(1-2^{-K_2})=2N(p-1).
\]
Hence $\Psi=1$. Weighting by physical distance gives
\[
 \Theta=\frac{\sum_{s=0}^{K_2-1}
             a_{\lfloor s/d\rfloor}2^{-s-1}}{1-1/p}.
\]
Grouping $s=dk+j$, $0\leq j<d$, and using
$\sum_{j=0}^{d-1}2^{-j-1}=1-2^{-d}$ gives~\eqref{eq:baselinef2}.

\subsection{SwingL}
Each of the $2d$ subcollectives sends $N/(2d)$ bits per node in each of
$K_2$ rounds, so every node sends $N$ bits in total per round. Consequently
\[
 \Vol_{\rm abs}=pNK_2,\quad
 \Psi=\frac{pdm}{2(p-1)},\quad
 \Theta=\frac{d\sum_ka_k}{dm}=\frac{A_m}{m}.
\]

\subsection{TrivanceB and TrivanceL on \texorpdfstring{$n=3^t$}{n=3\string^t}}
In reduce-scatter round $s$, each node sends two messages per subcollective,
each of size $N/(d3^{s+1})$, over distance $3^{\lfloor s/d\rfloor}$.
There are $d$ subcollectives. Thus
\[
 \Vol_{\rm abs}=4pN\sum_{s=0}^{K_3-1}3^{-s-1}
 =2pN(1-3^{-K_3})=2N(p-1),
\]
which proves $\Psi=1$, and
\[
 \Theta=\frac{2\sum_{s=0}^{K_3-1}
        3^{\lfloor s/d\rfloor}3^{-s-1}}{1-1/p}.
\]
Grouping $s=dk+j$ yields~\eqref{eq:baselinef3}.

For TrivanceL, each of the two messages contains the entire subvector, of size $N/d$ bits. The abstract communication volume is $2pNK_3$, hence
\[
 \Psi=\frac{pdt}{p-1},\qquad
 \Theta=\frac{1}{t}\sum_{k=0}^{t-1}3^k=\frac{n-1}{2t}.
\]
The factor two in the per-node traffic is essential: the same subvector
is sent to both peers, not divided between them.

For fixed $d>1$, the weighted sums also give
\[
 \lim_{m\to\infty}F_2(2^m,d)=
 \frac{2^d(2^d-1)}{(2^d-2)(2^d+1)},\qquad
 \lim_{t\to\infty}F_3(3^t,d)=
 \frac{3^d-1}{3^d-3}.
\]
These limits follow by summing geometric series. For $d=2$, they are $6/5$ and $4/3$, respectively, along the two different radix sequences.

\end{document}